\documentclass[preprint,11pt]{article}

\usepackage{amssymb,amsfonts,amsmath,amsthm,amscd,dsfont,mathrsfs}
\usepackage{graphicx,float,psfrag,epsfig,color,enumerate,hyperref}

\newcommand{\argmax}{\operatornamewithlimits{argmax}}

\newcommand\T{\rule{0pt}{2.6ex}}

\newtheorem{propo}{Proposition}[section]
\newtheorem{lemma}[propo]{Lemma}
\newtheorem{definition}[propo]{Definition}

\newtheorem{theorem}[propo]{Theorem}

\newtheorem{remark}[propo]{Remark}

\newtheorem{claim}[propo]{Claim}

\newcommand{\mc}{\mathcal}

\newcommand{\alf}[2]{\alpha_{#1\backslash #2}}

\newcommand{\soff}[2]{\widehat{m}_{#1\rightarrow #2}}
\newcommand{\off}[2]{m_{#1\rightarrow #2}}

\newcommand{\baralf}{\underline{\alpha}}

\newcommand{\bargamma}{\underline{\gamma}}

\newcommand{\bary}{\underline{y}}
\newcommand{\baroff}{\underline{m}}
\newcommand{\barsoff}{\underline{\widehat{m}}}

\newcommand{\surp}{\mathcal{S}\textnormal{urp}}

\newcommand{\sT}{{\sf{T}}}
\newcommand{\sI}{{\sf{I}}}

\def\reals{{\mathds R}}

\def\di{{\partial i}}

\def\cA{{\cal A}}

\def\Poly{{\rm Poly}}
\def\eps{{\epsilon}}

\def\bot{\textup{bot}}

\def\reb{\textup{\tiny reb}}
\def\thr{\textup{\tiny thr}}
\def\ext{\textup{\tiny ext}}

\def\damp{\kappa}

\def\u0{\underline{0}}

\def\poly{\textup{poly}}
\def\rfl{{\cal R}}

\makeatletter
\def\equalsfill{$\m@th\mathord=\mkern-7mu
\cleaders\hbox{$\!\mathord=\!$}\hfill
\mkern-7mu\mathord=$}
\makeatother

\definecolor{Red}{rgb}{1,0,0}
\definecolor{Blue}{rgb}{0,0,1}
\definecolor{Olive}{rgb}{0.41,0.55,0.13}
\definecolor{Green}{rgb}{0,1,0}
\definecolor{MGreen}{rgb}{0,0.8,0}
\definecolor{DGreen}{rgb}{0,0.55,0}
\definecolor{Yellow}{rgb}{1,1,0}
\definecolor{Cyan}{rgb}{0,1,1}
\definecolor{Magenta}{rgb}{1,0,1}
\definecolor{Orange}{rgb}{1,.5,0}
\definecolor{Violet}{rgb}{.5,0,.5}
\definecolor{Purple}{rgb}{.75,0,.25}
\definecolor{Brown}{rgb}{.75,.5,.25}
\definecolor{Grey}{rgb}{.5,.5,.5}
\definecolor{Pink}{rgb}{1,0,1}
\definecolor{DBrown}{rgb}{.5,.34,.16}

\definecolor{Black}{rgb}{0,0,0}

\begin{document}

\begin{titlepage}

\title{An FPTAS for Bargaining Networks with Unequal Bargaining Powers}

\author{Yashodhan Kanoria\footnote{Part of this work was done while the author was visiting Microsoft Research New England.
The author is supported by a 3Com Corporation Stanford Graduate Fellowship.}\\
Department of Electrical Engineering\\
Stanford University\\
Email: ykanoria@stanford.edu}

\date{}

\maketitle
\thispagestyle{empty}

\begin{abstract}
Bargaining networks model social or economic situations in which agents seek to form the most
lucrative partnership with another agent from among several alternatives.
There has been a  flurry of recent research studying Nash bargaining solutions (also called `balanced outcomes')
in bargaining networks, so that
we now know when such solutions exist, and also that they can be computed efficiently, even by market agents behaving
in a natural manner.

In this work we study a generalization of Nash
bargaining, that models the possibility of unequal `bargaining powers'.
This generalization was introduced in \cite{OurNewArxiv}, where it was shown that
the corresponding `unequal division' (UD) solutions exist if and only if
Nash bargaining
solutions exist, and also that a certain local dynamics converges to UD solutions when they exist.
However, the bound on convergence time obtained for that dynamics
was exponential in network size for the unequal division case.
This bound is tight, in the sense that there
exists instances on which the dynamics of \cite{OurNewArxiv}
converges only after exponential time.
Other approaches, such as the one of Kleinberg and Tardos, do not generalize to the unsymmetrical case.
Thus, the question of computational tractability of UD solutions has
remained open.

In this paper, we provide an FPTAS for the computation of UD solutions, when such solutions exist.
On a graph $G=(V,E)$ with weights (i.e. pairwise profit opportunities) uniformly bounded
above by $1$, our FPTAS finds an $\eps$-UD solution in time $\poly(|V|,1/\eps)$.
We also provide a fast \emph{local} algorithm for finding $\eps$-UD solution,
providing further justification that a market can find such a solution.
\end{abstract}

\end{titlepage}

\section{Introduction}
\label{sec:intro}

Bargaining networks  serve as a model
for various social or economic interactions where agents seek to form pairs for mutual benefit
(e.g. \cite{CookY,NET,Lucas}).
Situations which can be modeled as such include a housing market with buyers and sellers, a job market with job seekers
and employers, or individuals seeking to form relationships and pair up. Bargaining networks are also referred to in
the literature
as `assignment markets' \cite{Rochford} or `exchange networks' \cite{Skvoretz,KT}.

A bargaining network is an undirected graph, with weights on the edges representing
potential profits if the corresponding pair of agents `trade' with each other (see Section
\ref{subsec:model} for formal definitions). Profit from
a trade is split between the participating agents as per a mutual agreement. Agents are
constrained on the number of trades they can participate in. A natural postulate
in this setting is that an outcome should be \textit{stable}, i.e. no pair of agents
should be able to do better by each abandoning a current partner and trading
with each other instead.
The solution concept of `balanced
outcomes' \cite{Rochford,CookY,KT} postulates further that each pair of agents that trade
must play the pairwise Nash bargaining solution \cite{Nash}, given
the behavior of the rest of the network. Thus, the `edge surplus' (cf. Eq.~\eqref{eq:surplus_defn}),
or the excess over the sum of `best alternatives'
for each of the two parties, is postulated to be split equally. This is called the \textit{balance} condition. 

However, it is natural to expect that such symmetry
is rare in practice, and that some players tend to have greater `bargaining power' than others. Such
bargaining power can arise due to a variety of reasons. For example, a more patient player has
more bargaining power, all else being equal. This
phenomenon is well known in the Rubinstein game
\cite{sequential_bargaining} where nodes alternately make
offers to each other until an offer is accepted -- the node
with less time discounting earns more in the subgame perfect Nash equilibrium.

Empirical findings confirm this. A recent experimental study of such networks \cite{Kearns_exp}
found that individual differences played a part in determining outcomes, including the observation that
patience correlated positively with earnings.
A previous study even estimates and `corrects' for the effects of particular subject pairs to better uncover
network structure effects \cite{Skvoretz}.
This leads us to ask if the concept of `balanced outcomes' can be
suitably generalized to account for such asymmetry. It turns out that there is, in fact a simple generalization
to the unsymmetrical case.
Our previous work \cite{OurNewArxiv} introduced the generalized concept of unsymmetrical
`unequal division'  (UD) solutions,
and also characterized the existence of such solutions.

Somewhat surprisingly, the various algorithms devised to compute solutions in the symmetric
setting fail to generalize to the unequal division setting (see also Section \ref{subsubsec:cooperative_game_connections}).
For example, the algorithm of Kleinberg and Tardos
\cite{KT} proceeds via
a sequence of linear programs that maximize the minimum `slack'.
This does not seem to have a simple generalization
to the asymmetric case.
Thus, the question of computational tractability of solutions for the unsymmetrical case in bargaining networks has been
open.

Besides computational tractability, another important question is ``Can a market find the solution concept on its own?"
In this line of work, one looks for simple, local mechanisms that converge to a solution concept.
Azar et al \cite{Azar} proposed such a local mechanism
for the bargaining networks problem.
The convergence result in that work showed an exponential bound on convergence time for
the symmetric case. Also, it does not generalize to the unsymmetrical case.
Our recent work on local dynamics in bargaining networks \cite{OurNewArxiv}
introduced a new analysis technique that provides a proof of convergence even for the unsymmetrical case,
and a polynomial bound on convergence to an approximate solution for the symmetrical case. However, a crucial issue
(see Section \ref{sec:stability_critical} of this paper)
led to a worst case exponential time to convergence in the unsymmetrical case.
In this paper we resolve this issue, providing a new efficient local algorithm for the unsymmetrical case.

\paragraph{Contributions.}
This work makes the following contributions in the context of bargaining networks:
\begin{itemize}
\item We establish computational tractability for bargaining networks with unequal bargaining
powers by providing
the first FPTAS for the corresponding `unequal division' solutions.
\item We provide a simple local algorithm and show that it converges fast to approximate unequal
division solutions. Specifically, it is a two phase algorithm: (i) The first phase consists of finding
the maximum weight matching and a stable allocation using belief propagation \cite{Bayati}.
(ii) The second phase consists of unsymmetrical edge balancing of the allocation, converging
to an approximate solution in polynomial time.
\end{itemize}

We note that the local algorithm we provide is similar to the one given by Azar et al \cite{Azar}
for the symmetrical case. In that work also there is a phase of matching using belief propagation, followed by
a phase of
edge balancing. However, several critical differences in both the design and the analysis of the algorithm enable us to
overcome limitations of their approach.

\subsection{Model}
\label{subsec:model}

A bargaining network consists of an undirected graph $G=(V,E)$ with positive weights on the edges,
denoted by $(w_e, e \in E) \in (0, W]^{|E|}$ (where $W>0$ denotes an arbitrary bound on weights).
Edges represent potential `trades', and weights
are the corresponding `profits'. Players are constrained on the number of trades they are allowed to participate in.
For simplicity, we will work with the \textit{one exchange rule}, i.e. each player is allowed to
participate in at most one trade. All our results easily generalize to the case of arbitrary
integral constraints on number of trades for each player.

If a pair of players trade with each other, the profit must be divided between them. Thus, a
\emph{trade outcome} or just an \emph{outcome} consists of a matching $M$ between players,
and an \textit{allocation} $\bargamma\in \reals_+^{|V|}$ such that $\gamma_i+\gamma_j=w_{ij}$
for each pair $(i,j) \in M$, and for each node $k\in V$ that is unmatched under $M$, $\gamma_k=0$.

Given a trade outcome $(\bargamma,M)$, we
define implicit \emph{offers} on all edges not in
$M$. Let $(x)_+ \equiv \max(x,0)$. For any $(i,j) \in E \backslash M$, node $i$ offers node $j$ an amount $(w_{ij} - \gamma_i)_+$, the idea being that
$i$ should be willing to switch partners if she can earn even slightly more. Thus, each node has a set
of well defined `alternatives' to its current partner in $M$. A natural postulate is that an outcome should
be \emph{stable}, i.e. for each node $i$, $\gamma_i$ should be no smaller than the best alternative
of node $i$ (if $i$ is unmatched under $M$, she should receive no
non-zero offers). The stability condition can be concisely written as $\gamma_i + \gamma_j \geq w_{ij}$
for all $(i,j) \in E \backslash M$.

Let $\partial i$ denote the set of neighbors of node $i$ in $G$.
For each edge $(ij) \in M$,  we define the
`edge surplus' as the excess of $w_{ij}$ over the sum of best alternatives, i.e.,
\begin{align}
\surp_{ij}(\bargamma) =  w_{ij} - \max_{k \in \partial i \backslash j} (w_{ik} - \gamma_k)_+
- \max_{l \in \partial j \backslash i} (w_{jl} - \gamma_l)_+ \, .
\label{eq:surplus_defn}
\end{align}
We can think of each node in the network as having an inherent `bargaining power', such that
  $\surp_{ij}$ should be split in a manner determined by the
bargaining powers of $i$ and $j$.
 We adopt a general model
where the surplus is postulated to be split as per a fraction $r_{ij} \in (0,1)$ for each matched edge $(ij) \in M$.
We call this \emph{correct division}.
Each $r_{ij}$ can be an arbitrary number in the interval $(0,1)$, independently for all edges.

\begin{definition}
A \emph{problem instance} $I$ consists
of an undirected graph $G=(V,E)$, with positive weights $(w_e)_{e\in E}$ and split fractions $(r_{ij})_{(ij) \in E} \in (0,1)^{|E|}$.
An arbitrary direction is chosen on each edge for purposes of specifying the split fraction.
If $r_{ij}$ is specified, then it is implicit that $r_{ji}=1-r_{ij}$.
\end{definition}

\begin{definition}[Correct division]
\label{def:correct_div}
An outcome $(\bargamma, M)$ is said to satisfy \emph{correct division} if, for all $(ij) \in M$,
\begin{align}
\gamma_i &= \max_{k \in \partial i \backslash j} (w_{ik} - \gamma_k)_+ + r_{ij} \surp_{ij}
\label{eq:correct_division}
\end{align}
where $r_{ji}=1-r_{ij}$ and $\surp_{ij}$ is defined by Eq.~\eqref{eq:surplus_defn}.
\end{definition}
\noindent Note that it follows from Eq.~\eqref{eq:correct_division} and Eq.~\eqref{eq:surplus_defn} that
$\gamma_j =w_{ij}-\gamma_i=  \max_{l \in \partial j \backslash i}\, (w_{jl} - \gamma_l)_+ + r_{ji} \surp_{ij}$.

\begin{definition}
\label{def:UD_soln}
An outcome $(\bargamma,M)$ is said to be an \emph{unequal division (UD) solution} if it is \emph{stable} and satisfies
\emph{correct division} (cf. Definition \ref{def:correct_div}).
\end{definition}

\subsection{Related work}
\label{subsec:related_work}

We present here a short review of relevant related work.

Recall the linear programming relaxation of the maximum weight matching problem
\begin{eqnarray}
\textup{maximize} && \sum_{(i,j) \in E} w_{ij} x_{ij},\nonumber\\
\textup{subject to}&&
\sum_{j\in \di} x_{ij} \le 1 \;\;\; \forall i \in V,\;\;\;\;\;\;\;
x_{ij}\ge 0\;\;\;\forall (i,j) \in E\,.
\label{prob:mwm_relaxation}
\end{eqnarray}

The dual problem to \eqref{prob:mwm_relaxation} is:
\begin{eqnarray}
\textup{minimize} && \sum_{i \in V} y_i,\nonumber\\
\textup{subject to}&&
y_i+y_j \ge w_{ij} \;\;\; \forall (i,j) \in E,\;\;\;\;\;\;\;
y_i \ge 0 \;\;\;\forall i \in V
\label{prob:mwm_dual}
\end{eqnarray}

Sotomayor \cite{Sotomayor} characterized the existence of stable outcomes in exchange networks.
\begin{lemma}[\cite{Sotomayor,KT}]
\label{lemma:stable_outcomes}
Stable outcomes exist if and only if the LP \eqref{prob:mwm_relaxation} has an integral optimum.
Further, if $(\bargamma,M)$ is a stable outcome, then $\bargamma$ is an optimum solution
of the dual LP \eqref{prob:mwm_dual} and $M$ is a maximum weight matching. Conversely, if the
LP \eqref{prob:mwm_relaxation} has an integral optimum, then for any maximum weight matching $M^*$
and any optimum $\bary^*$ of the dual LP \eqref{prob:mwm_dual}, $(\bary^*,M^*)$ is a stable outcome.
\end{lemma}

The above lemma follows from the stability condition $\gamma_i+\gamma_j \geq w_{ij}$ for all $(ij) \in M$.
It implies, in particular, that all instances on bipartite graphs possess stable outcomes.

There have been several recent works on the symmetrical `balanced outcome' solution concept (corresponding to $r_{ij}=1/2$
for all $(ij) \in E$),
following a paper by Kleinberg and Tardos \cite{KT,Azar,Bateni,OurNewArxiv}.

Though our previous work \cite{OurNewArxiv} focuses
on the symmetrical case, it also introduces unequal division solutions. Further, it shows that unequal division
solutions exist if and only if Nash bargaining solutions exist.

\begin{theorem}[\cite{OurNewArxiv}]
\label{thm:UD_existence}
A problem instance admits a UD solution if and only if it admits a stable outcome (which occurs iff the LP \eqref{prob:mwm_relaxation} has an integral optimum).
\end{theorem}
This generalizes a result of Kleinberg and Tardos for existence of balanced outcomes \cite{KT}.

\cite{OurNewArxiv} also shows that a certain local dynamics converges to UD solutions, when such solutions exist. However, the bound
on time to convergence is exponential in the network size (in contrast to the symmetrical case),
and this bound turns out to be tight in worst case (see Section \ref{sec:stability_critical}). Here, we resolve this issue, providing a
new FPTAS
for computing approximate UD solutions.

\subsubsection{Relationship to Cooperative games}
\label{subsubsec:cooperative_game_connections}
Recent work by Bateni et al
\cite{Bateni} shows that the bargaining network setting can be viewed as a cooperative game, making
this problem susceptible to a large body of literature (see also \cite{Rochford}).
This literature defines various solution concepts such as nucleolus, kernel and prekernel, and
also investigates means to compute these solutions for various classes of games.\footnote{\cite{Bateni}
shows that stable, balanced outcomes in bargaining networks
correspond to the core intersection prekernel.}
It is noteworthy that all the solution concepts studied are symmetric in the players.
Whereas such concepts may form a reasonable predictive framework in the absence of player specific information,
we also want to ask ``Can the players reach an appropriate `solution' when there is asymmetry?" To this
end, we would like to establish computational tractability in the asymmetric case.

However, a little investigation
reveals that the approaches devised to compute various (symmetric) solution concepts rely heavily on the
symmetry in their respective definitions.
For instance, the polynomial time algorithm in Faigle et al \cite{Faigle01} for finding a point in the
least core intersection prekernel
uses two components --a transfer scheme
and a linear programming based update-- neither of which work in the unsymmetrical case.

The situation is similar with regard to simple `transfer schemes' that converge to a solution concept.
For
the general cooperative game problem, Maschler proposed a simple transfer scheme to approximate points in
the prekernel. A version of this
scheme was shown to converge by Stearns, and a simpler proof of convergence
was provided by Faigle et al \cite{Faigle01}, in the general cooperative game setting.
However, both proofs suffer from two drawbacks: (a) they depend on the symmetry of the solution concepts, (b) the bound
on convergence time (if any) is exponential in network size.  Essentially the same transfer scheme was used in Azar et al \cite{Azar}
for bargaining networks (see \cite{Bateni} for the connection), and the proof of convergence suffered from the same drawbacks.

The current work addresses computational tractability for the asymmetric case in the bargaining network setting,
where an appropriate asymmetric solution concept can be readily defined.

\subsection{Outline of the paper}
We present our FPTAS in Section \ref{sec:main_results}, along with a proof that it returns an $\eps$-UD solution in polynomial time.
We present a fast local algorithm for this problem in subsection \ref{subsec:local_algorithm}. Each of the algorithms involve an
iterative `rebalancing' phase.
Section \ref{sec:lemma_proofs} contains proofs of
some key Lemmas used. In Section \ref{sec:stability_critical}, we demonstrate the importance of ensuring
that we stay within the subset of \emph{stable} allocations in our iterative updates. This insight is critically
used in our construction of an FPTAS. Appendix \ref{app:dualopt_from_Mstar}
shows a polynomial time local `reduction' from the problem of finding an $\eps$-UD solution to the problem of finding a maximum weight matching.

\section{Main results}
\label{sec:main_results}

First we define an approximate version of correct division, asking that
Eq.~\eqref{eq:correct_division} be satisfied to within an additive $\eps$, for all matched edges.
\begin{definition}[$\eps$-Correct division]
\label{def:eps_correct_div}
An outcome $(\bargamma, M)$ is said to satisfy \emph{$\eps$-correct division} if, for all $(ij) \in M$,
\begin{align}
|\gamma_i - \max_{k \in \partial i \backslash j} (w_{ik} - \gamma_k)_+ - r_{ij} \surp_{ij}(\bargamma)| \leq \eps
\label{eq:eps_correct_div}
\end{align}
where $\surp_{ij}(\cdot)$ is defined by Eq.~\eqref{eq:surplus_defn}.
\end{definition}

We define approximate UD solutions as follows:
\begin{definition}
An outcome $(\bargamma,M)$ is an \emph{$\eps$-UD solution} for $\eps \geq 0$
if it is \emph{stable} and it satisfies \emph{$\eps$-correct division} (cf. Definition \ref{def:eps_correct_div}).
\end{definition}

This is analogous to the definition of $\eps$-Nash equilibrium (see, e.g. \cite{approx_NE}).

It follows from Lemma \ref{lemma:stable_outcomes} that $\eps$-UD solutions exist iff the LP
\eqref{prob:mwm_relaxation} admits an integral optimum. This is the same as the requirement for
existence of UD solutions (see Theorem \ref{thm:UD_existence}). Our main result is the following:
\begin{theorem}
\label{thm:FPTAS}
There is a $\Poly(|V|, 1/\eps)$ algorithm such that for any problem instance with weights uniformly bounded by $1$,
i.e. $(w_e, e \in E) \in (0,1]^{|E|}$:
\begin{itemize}
\item If the instance admits a UD solution, the algorithm finds an $\eps$-UD solution.
\item If the instance does not admit a UD solution the algorithm
returns a message {\sc unstable}.
\end{itemize}
\end{theorem}

Our approach to finding an $\eps$-UD solution consists of two main steps:
\begin{enumerate}
\item Find a maximum weight matching $M^*$ and a dual optimum $\bargamma$ (solution to the dual LP \eqref{prob:mwm_dual}) .
Thus, form a stable outcome $(\bargamma,M^*)$. Else certify that the instance has no UD solution.
\item Iteratively update the allocation $\bargamma$ without changing the matching.
Updates are local, and are designed to converge fast to an allocation satisfying the $\eps$-correct division solution
\emph{while maintaining stability}. Thus, we arrive at an $\eps$-UD solution.
\end{enumerate}
As mentioned earlier, this is similar to the approach of \cite{Azar}. The crucial differences (enabling our results)
are:  (i) we stay within the space of stable outcomes, and (ii) our analysis of convergence.

First let us focus on obtaining an FPTAS using the steps above. Later we describe how to make the algorithm local.

Step 1 can be carried out by finding a maximum weight matching $M^*$ (e.g. \cite{PolyMWM}) and also solving the
the dual linear program \eqref{prob:mwm_dual}. For the dual LP, let $\cal{V}$ be the optimum value and let
$\bargamma$ be an optimum solution. We now use Lemma \ref{lemma:stable_outcomes}.
If the weight of $M^*$ is smaller than $\cal{V}$, we return {\sc unstable}, since we know
that no stable outcome exists, hence no UD solution (or $\eps$-UD solution) exists. Else,
$(\bargamma,M^*)$ is a stable outcome. This completes step 1! The computational effort involved is $\poly(|V|)$.

In step 2, we fix the matching $M^*$, and rebalance the edges iteratively. It turns out to be crucial that our iterative
updates preserve stability.  Section \ref{sec:stability_critical} demonstrates that the rebalancing procedure can take an
exponentially large time to reach an approximate UD solution if stability is not preserved. 

We motivate the rebalancing procedure
briefly, before we give a detailed description and state results. Imagine an edge $(i,j) \in M^*$. Since we start with a
stable outcome, the edge weight $w_{ij}$ is at least the sum of the best alternatives, i.e. $\surp_{ij}\geq0$. Suppose we change the
division of $w_{ij}$ into $\gamma_i'$, $\gamma_j'$ so that the $\surp_{ij}$ is divided as per the prescribed split
fraction $r_{ij}$. Earnings of all other nodes are left unchanged.
Since $r_{ij} \in (0,1)$, $\gamma_i'$ is at least as large as the best alternative of $i$, as was the case for $\gamma_i$.
This leads to $\gamma_i'+\gamma_k \geq w_{ik}$ for all $k \in \partial i \backslash j$. A similar argument holds for node $j$.
In short, \emph{stability is preserved}!

It turns out that the analysis of convergence is simpler if we analyze synchronous updates, as opposed to
asynchronous updates as described above. Moreover, we find that simply choosing an appropriate `damping factor' allows us to ensure that stability
is preserved even with synchronous updates. We use a powerful technique introduced in our recent work \cite{KT} to prove
convergence.


Table \ref{alg:edge_reb} shows the algorithm {\sc Edge Rebalancing} we use to complete step 2.
Note that each iteration
of the loop can requires $O(|E|)$ simple operations.\\

\begin{table}[t]
\caption{Local algorithm that converts stable outcome to $\eps$-UD solution}
\begin{tabular}{ll}
\hline
\multicolumn{2}{l}{ \T {\sc Edge Rebalancing}( Instance $I$, Stable outcome $(\bargamma,M)$, Damping factor $\damp$, Error target $\eps$)}\\[2pt]
\hline
1: & Check $\damp \in (0, 1/2]$, $\eps > 0$, $(\bargamma,M)$ is stable outcome \T \\[2pt]
2: & If (Check fails)\hspace{0.1cm} Return {\sc error}\\[2pt]
3: & $\bargamma^0 \leftarrow \bargamma$\\[2pt]
4: & $t \leftarrow 0$\\[2pt]
5: & Do\\[2pt]
6: & \hspace{0.4cm} ForEach $(i,j) \in M$\\[2pt]
7: & \hspace{0.9cm} $\gamma^{\reb}_i \leftarrow \max_{k \in \partial i \backslash j} (w_{ik} - \gamma_k^t)_+ + r_{ij} \surp_{ij}(\bargamma^t)$\\[2pt]
8: & \hspace{0.9cm} $\gamma^{\reb}_j \leftarrow \max_{l \in \partial j \backslash i}\, (w_{jl} \, - \gamma_l^t)_+ + r_{ji} \surp_{ij}(\bargamma^t)$\\[2pt]
9: & \hspace{0.4cm} End ForEach\\[2pt]
10: & \hspace{0.4cm} ForEach $i\in V$ that is unmatched under $M$\\[2pt]
11: & \hspace{0.9cm} $\gamma^{\reb}_i \leftarrow 0$\\[2pt]
12: & \hspace{0.4cm} End ForEach\\[2pt]
13: & \hspace{0.4cm} If $\left(\lVert\bargamma^{\reb} -\bargamma^t\rVert_\infty \leq \eps \right)$\hspace{0.1cm} Break Do \\[2pt]
14: & \hspace{0.4cm} $\bargamma^{t+1} =  \damp \bargamma^{\reb} + (1-\damp) \bargamma^t$\\[2pt]
15: & \hspace{0.4cm} $t \leftarrow t+1$ \\[2pt]
16: & End Do \\[2pt]
17: & Return $(\bargamma^t, M)$\\[2pt]
\hline
\end{tabular}
\label{alg:edge_reb}
\end{table}

\noindent{\bf Correctness  of {\sc Edge Rebalancing}:}\\
A priori, it is not clear that $\bargamma^t$ computed by {\sc Edge Rebalancing} is a stable allocation (or even
an allocation) corresponding to $M$, for $t>0$.
The following lemma eliminates this concern.

\begin{lemma}
\label{lemma:ER_preserves_stability}
If {\sc Edge Rebalancing} is given a valid input satisfying the `Check' on line 1, then
$(\gamma^t, M)$ is a stable outcome for all $t \geq 0$.
\end{lemma}
This guarantees that  {\sc Edge Rebalancing} returns an $\eps$-UD solution if it terminates (unless it returns {\sc error}).
The lemma is a straightforward consequence
of the constraint $\damp \leq 1/2$ (proof in Section \ref{sec:lemma_proofs}).

\noindent{\bf Convergence  of {\sc Edge Rebalancing}:}\\
Next we need to show that the rebalancing algorithm terminates fast at an $\eps$-UD solution.
Note that the termination condition $\lVert\bargamma^{\reb} -\bargamma^t\rVert_\infty \leq \eps$ on Line 13 is
equivalent to $\eps$-correct division.
%

\begin{lemma}
\label{lemma:ER_converges_fast}
For any instance with weights bounded by $1$, i.e. $(w_e, e \in E) \in (0,1]^{|E|}$, if {\sc Edge Rebalancing}
is given a valid input, it terminates
in $T$ iterations, where
\begin{align}
T \leq \left \lceil \frac{1}{\pi \damp (1-\damp) \eps^2} \right \rceil \, ,
\label{eq:termination_time_bound}
\end{align}
and returns an outcome satisfying $\eps$-correct division (cf. Definition \ref{def:eps_correct_div}). Here $\pi = 3.14159\ldots$
\end{lemma}

The proof is in Section \ref{sec:lemma_proofs}.

Using Lemmas \ref{lemma:ER_preserves_stability} and Lemmas \ref{lemma:ER_converges_fast}, we immediately obtain
our main result, Theorem \ref{thm:FPTAS}.

\begin{proof}[Proof of Theorem \ref{thm:FPTAS}]
We showed that step 1 can be completed in time $\poly(|V|)$.
If the instance has no UD solutions then the algorithm returns {\sc unstable}. Else
we obtain a stable outcome and proceed to step 2.

Step 2 is performed using {\sc Edge Rebalancing}. The input is the instance, the stable outcome obtained from step 1,
$\damp = 1/2$ (for example) and the target error value $\eps > 0$. Lemmas \ref{lemma:ER_preserves_stability} and \ref{lemma:ER_converges_fast}
show that {\sc Edge Rebalancing} terminates after at most $\lceil 1/(\pi \damp (1-\damp) \eps^2) \rceil$ iterations, returning a outcome that is
stable and satisfies $\eps$-correct division, i.e. an $\eps$-UD solution.
Moreover, each iteration requires $O(|E|)$ simple operations. Hence, step 2 is completed in  $O(|E|/\eps^2)$ simple operations.

The total number of operations required by the entire algorithm is thus $\poly(|V|,1/\eps)$.
\end{proof}

\subsection{A Fast Local Algorithm}
\label{subsec:local_algorithm}

Our algorithm {\sc Edge Rebalancing} for step 2  is local/distributed, with each matched edge in the
graph being updated according to the same, time invariant rule.
This rule is a simple function of the edge parameters (weight, split fraction), and the current earnings of nodes in the
1-hop neighborhood. Only the termination condition is centrally computed,
but even that \emph{can be replaced with fixed time} $T_* = \left \lceil 1/(\pi \damp (1-\damp) \eps^2 \right \rceil$
at which to terminate (cf. Section \ref{sec:lemma_proofs}, Remark \ref{rem:ER_fixed_termination_time_possible}). Note that
$T_*$ is \emph{independent of network size}. It is also worth mention  that since stability is preserved, no
player ever has incentive to change her partner. Thus, {\sc Edge Rebalancing} constitutes a plausible model for
behavior of market participants, after they have attained a stable outcome. Damping can be interpreted as inertia
to change in the status quo.

We now show that step 1 can also be accomplished by a fast local algorithm, when
the LP \eqref{prob:mwm_relaxation} has a unique optimum (this condition is generic, see
Remark \ref{rem:unique_opt_isgeneric}, Appendix \ref{app:fast_local}).

The local algorithm
we use is belief propagation for maximum weight matching \cite{Bayati,BayatiB,SMW07,JHu07}.
This is a message passing algorithm with iterative updates of a `message' vector $\baroff \in [0,W]^{2|E|}$.
There are two messages on each edge $(i,j)$, denoted by $\off{i}{j}$ and  $\off{j}{i}$, one in each direction.
The algorithm performs iterative updates according to
\begin{align}
\off{i}{j}^{t+1} = \left(w_{ij} - \max_{k \in \partial i \backslash j} \off{k}{i}^t \right)_+
\label{eq:BP_one_step_update}
\end{align}
The initialization is the all-zero message vector $\baroff^0 = \underline{0}$. We denote this algorithm by BP-MWM.

\begin{lemma}[\cite{Bayati,BayatiB,SMW07}]
\label{lemma:BP_bayati_etal}
Suppose LP \eqref{prob:mwm_relaxation} has a unique optimum. BP-MWM converges iff the
optimum is integral. Further, if LP \eqref{prob:mwm_relaxation}
has an integral optimum corresponding to matching $M^*$, then the messages converge in $\left \lceil\frac{2|V|W}{g}\right \rceil $
iterations to a fixed point $\baroff^*$ satisfying the following. For any $i \in V$, if $i$ is matched under $M^*$ to $j$, then
$\argmax_{k \in \partial i} \off{k}{i}^* = \{j\}$ and $\off{j}{i}^*>0$. If $i$ is unmatched under $M^*$ then $\off{k}{i}=0$ for every $k \in \partial i$.
\end{lemma}

Note how the condition for convergence of BP-MWM is the same
as that for the existence of UD solutions! Here $g$ is difference in the weights of the heaviest and next heaviest corner of the matching polytope (given by the constraints of
LP \eqref{prob:mwm_relaxation}). We call it the \emph{LP gap}.

In the case that LP \eqref{prob:mwm_relaxation} has a unique optimum,
BP-MWM thus answers ``Does the LP \eqref{prob:mwm_relaxation} have an integral optimum?" If yes, it
also finds the maximum weight matching.

In fact, one also can find an optimum solution to the dual LP \eqref{prob:mwm_dual}
from $\baroff^*$, when BP-MWM converges (see also \cite[Appendix F.1]{OurNewArxiv}).
Consider any $i \in V$. Sort the messages $(\off{k}{i}^* , k \in \partial i)$ is non-increasing order. Denote the value of the
first item in the sorted list by $\mu_i(1)$ and the next value by $\mu_i(2)$.
Define $y_i^* \equiv (\mu_i(1)+\mu_i(2))/2$. The following
is proved in Appendix \ref{app:fast_local}.

\begin{propo}
\label{propo:BP_gives_dual_opt}
The construction above produces $\bary^*$, an optimal solution to the dual LP \eqref{prob:mwm_dual}.
\end{propo}

Thus, we obtain a stable outcome $(\bary^*,M^*)$ from the BP fixed point $\baroff^*$ (see Lemma \ref{lemma:stable_outcomes}).

We mention here that BP-MWM can be interpreted as a bargaining process \cite[Appx A]{Azar}.

\begin{remark}
The performance of BP-MWM seems to be lacking in two respects. First, it fails when the LP \eqref{prob:mwm_relaxation}
has an integral optimum that is not unique. Second, the bound on convergence time depends inversely on LP gap $g$ which may be
arbitrarily small (in fact the bound is tight in worst case). We make three comments on this issue:
\begin{enumerate}[(i)]
\item These `flaws' appear to be inevitable. We are not aware of any local algorithm for maximum
weight matching that overcomes them.
\item For any instance on a \emph{bipartite} graph, the LP gap $g$ is larger than inverse polynomial in $|V|$ with
probability close to 1 under small random perturbations  \cite[Lemma 1]{OurNewArxiv}.
Thus, BP-MWM is likely to converge in time $\poly(|V|)$ on bipartite graphs as per this `smoothed analysis'.
\item Appendix \ref{app:dualopt_from_Mstar}
shows that if we are given a maximum weight matching
$M^*$ for an instance possessing a UD solution, then we can locally construct a stable outcome in $\poly(|V|)$ operations.
Using this, we obtain a \emph{local polynomial time reduction} from the problem of finding an $\eps$-UD solution to the sub-problem
of finding a maximum weight matching.
\end{enumerate}
\end{remark}

\section{Proofs of Lemmas \ref{lemma:ER_preserves_stability} and \ref{lemma:ER_converges_fast}}
\label{sec:lemma_proofs}

\begin{proof}[Proof of Lemma \ref{lemma:ER_preserves_stability}]
We prove this lemma by induction on time $t$. Clearly $(\gamma^0, M)$ is a stable outcome, since the input is valid. Suppose $(\gamma^t, M)$ is a stable outcome.

Consider any $(i,j) \in M$. It is easy to verify that $\gamma_i^\reb + \gamma_j^\reb = w_{ij}$, for
$\bargamma^{\reb}$ computed from $\bargamma^t$ in
Lines 8-11 of {\sc Edge Rebalancing}. Also, we know that $\gamma_i^t + \gamma_j^t = w_{ij}$.
It follows that $\gamma_i^{t+1} + \gamma_j^{t+1} = w_{ij}$ as needed. For $i \in V$ unmatched
under $M$, $\gamma_i^t=0$ by hypothesis and $\gamma_i^\reb=0 \, \Rightarrow \, \gamma_i^{t+1}=0$ as needed.

Consider any $(i,k) \in E \backslash M$. We know that $\gamma_i^t + \gamma_k^t \geq w_{ik}$. We want to show the corresponding inequality
at time $t+1$. Define $\sigma_{ik}^t \equiv \gamma_i^t + \gamma_k^t - w_{ik} \geq 0$.

{\bf Claim:} $\gamma_i^{\reb} \geq \gamma_i^t - \sigma_{ik}^t$

If we prove the claim, it follows that a similar inequality holds for $\gamma_k^{\reb}$, and hence
$\gamma_i^{\reb}+\gamma_k^{\reb} \geq\gamma_i^t + \gamma_k^t - 2\sigma_{ik}= w_{ik}- \sigma_{ik}^t$.
It then follows from the definition in Line 14
that $\gamma_i^{t+1} + \gamma_k^{t+1} \geq w_{ik}$,
for any $\damp \in (0,1/2]$. This will complete our proof that $(\bargamma^{t+1}, M)$ is a stable outcome.

Let us now prove the claim. Suppose $i$ is matched under $M$.
Using the definition in Line 7 (Line 8 contains a symmetrical definition), 
$\gamma_i^\reb \geq \max_{k' \in \partial i \backslash j} (w_{ik'} - \gamma_{k'}^t)_+$ since $\surp_{ik}(\bargamma^t) \geq 0$. Hence,
\begin{align*}
\gamma_i^\reb \geq (w_{ik} - \gamma_k^t)_+ \geq (w_{ik} - \gamma_k^t) = \gamma_i^t - \sigma_{ik}^t
\end{align*}
as needed. If $i$ is not matched under $M$, then $\gamma_i^t=\gamma_i^\reb=0$, so the claim follows from $\sigma_{ik}^t \geq 0$.
\end{proof}

\begin{proof}[Proof of Lemma \ref{lemma:ER_converges_fast}]
This result is proved using the powerful technique introduced in our recent work \cite{OurNewArxiv}.
We show that the iterative updates of {\sc Edge Rebalancing} can be written as
\begin{align}
\bargamma^{t+1} = \damp \sT\bargamma^t + (1-\damp)\bargamma^t
\label{eq:mann_iterations}
\end{align}
where $\sT$ is a non-expansive self mapping of a bounded convex subset of a normed linear space.

The linear space we consider is simply $\reals_+^{|V|}$. Let $\cA_M \subseteq [0,W]^{|V|}$ be the
set of allocations corresponding to matching $M$. It is easy to see that $\cA_M$
is a bounded convex set. We define $\sT: \cA_M \rightarrow \cA_M$
as the product of two operators, `rebalancing' operator $\sT^{\reb}: \cA_M \rightarrow \cA_M^{\ext}$
and a `thresholding' operator $\sT^{\thr}: \cA_M^{\ext} \rightarrow \cA_M$. Here $\cA_M^{\ext} \supseteq \cA_M$ is
set of allocations corresponding to matching $M$, with the non-negativity constraint relaxed. We define
$\sT^{\reb}$ as follows. For each $i \in V$ that is unmatched under $M$, $(\sT^\reb\bargamma)_i\equiv 0$.
For each $(i,j) \in M$,
\begin{align}
(\sT^{\reb} \bargamma)_i &\equiv \max_{k \in \partial i \backslash j} (w_{ik} - \gamma_k)_+ + r_{ij} \surp_{ij}(\bargamma)  \label{eq:Treb_defined} \\
(\sT^{\reb} \bargamma)_j &\equiv w_{ij} - (\sT^{\reb} \bargamma)_i = \max_{l \in \partial j \backslash i}\, (w_{jl} - \gamma_l)_+ + r_{ji} \surp_{ij}(\bargamma)
\end{align}
Note that $\bargamma^\reb$ as defined in Lines 6-12 of {\sc Edge Rebalancing} is exactly $\sT^{\reb} \bargamma^t$.
Also note that $\sT^{\reb} \bargamma \in \cA_M^{\ext}$ as required.

We define $\sT^{\thr}$ as follows. For each $i \in V$ that is unmatched under $M$, $(\sT^\thr \bargamma)_i\equiv 0$.
For each $(i,j) \in M$, there are three cases.\\
$\gamma_i <0$: $(\sT^{\thr} \bargamma)_i \equiv 0, (\sT^{\thr} \bargamma)_j \equiv w_{ij}$\\[2pt]
$\gamma_j <0$: $(\sT^{\thr} \bargamma)_i \equiv w_{ij}, (\sT^{\thr} \bargamma)_j \equiv 0$\\[2pt]
$\gamma_i \geq 0, \gamma_j \geq 0$: $(\sT^{\thr} \bargamma)_i \equiv \gamma_i, (\sT^{\thr} \bargamma)_j \equiv \gamma_j$
\vskip3pt

Note that $\gamma_i < 0$ and $\gamma_j < 0$ cannot occur simultaneously since $\gamma_i + \gamma_j = w_{ij}$.
Also note that the result of operating with $\sT^{\thr}$ is in $\cA_M$.

Consider the composite operator $\sT \equiv \sT^{\thr} \sT^{\reb}$. If $\sT$ operates on a stable outcome, the output of operator
$\sT^{\reb}$ is a non-negative allocation (since $\surp_{ij} \geq 0$ for every $(i,j) \in M$) with earnings of unmatched nodes being $0$,
 and $\sT^{\thr}$ acts simply as an
identity operator. It follows (using Lemma \ref{lemma:ER_preserves_stability}) that Lines 6-12 define $\bargamma^{\reb}= \sT \bargamma^t$.
Thus, we have verified that the iterative updates of {\sc Edge Rebalancing} (Line 14) correspond to Eq.~\eqref{eq:mann_iterations}.

Next, we show that $\sT$ is non expansive in sup norm, i.e. for any $\bargamma^a, \bargamma^b \in \cA_M$,
\begin{align}
\label{eq:T_non_expansive}
\lVert \sT \bargamma^a - \sT \bargamma^b \rVert _\infty \leq \lVert \bargamma^a - \bargamma^b \rVert _\infty
\end{align}
We prove this by showing that each of $\sT^{\reb}$ and $\sT^{\thr}$ is non-expansive in sup norm.

Consider $\sT^{\reb}$. Take any $(i,j) \in M$. Rewriting Eq.~(\ref{eq:Treb_defined})
using Eq.~\eqref{eq:surplus_defn},
we have
\begin{align*}
(\sT^{\reb} \bargamma)_i = r_{ij}w_{ij}+ (1-r_{ij}) \max_{k \in \partial i \backslash j} (w_{ik} - \gamma_k)_+
    - r_{ij} \max_{l \in \partial j \backslash i}\, (w_{jl} - \gamma_l)_+
\end{align*}
Now $x \mapsto (w-x)_+$ is non-expansive, and the `max' operator is non-expansive. Hence, using the triangle inequality
we obtain
\begin{align*}
|(\sT^{\reb} \bargamma^a)_i - (\sT^{\reb} \bargamma^b)_i| \leq \lVert \bargamma^a - \bargamma^b \rVert _\infty
\end{align*}
and similarly for $j$. For each $k$ that is unmatched under $M$, $(\sT^{\reb} \bargamma)_k =0$.
It follows that $\sT^{\reb}$ is non-expansive in sup norm.

Next consider $\sT^{\thr}$. For each $k$ that is unmatched under $M$, $(\sT^{\thr} \bargamma)_k =0$.
For any $(i,j) \in M$, we can write $(\sT^{\thr} \bargamma)_i = \max(\min(\gamma_i, w_{ij}), 0)$.
Since the `max'  and `min' operators are non-expansive, it follows that $\sT^{\thr}$ is non-expansive in sup norm.

Thus, we have shown that $\sT$ is a non-expansive self mapping of a bounded convex set of diameter 1 (since $W=1$).
Also, $\bargamma^t$ is obtained via iterative updates as per Eq.~\eqref{eq:mann_iterations}.
The main theorem in \cite{rate} tells us that
\begin{align}
\lVert \sT \bargamma^t - \bargamma^t \rVert _\infty \leq \frac{1}{\sqrt{\pi \damp (1-\damp) t}}
\label{eq:BB_rate}
\end{align}
Eq.~\eqref{eq:termination_time_bound} follows. Also,  $\lVert \sT^\reb\bargamma^{T} -\bargamma^T\rVert_\infty =
\lVert\bargamma^{\reb} -\bargamma^T\rVert_\infty \leq \eps$
implies $\eps$-correct division for $(\bargamma^T,M)$.
\end{proof}

\begin{remark}
\label{rem:ER_fixed_termination_time_possible}
In light of Eq.~(\ref{eq:BB_rate}), we could have simply used a fixed termination time of
$T_* = \lceil 1/(\pi \damp (1-\damp) \eps^2) \rceil$,
instead of the termination condition in Line 13. Eq.~(\ref{eq:BB_rate}) guarantees that
$\bargamma^{T_*}$ satisfies the $\eps$-correct division condition.
\end{remark}

\begin{remark}
\label{rem:Ishikawa_convergence}
It we remove the termination condition on Line 13 of {\sc Edge Rebalancing} (and iterate forever), \cite[Corollary
1]{Ishikawa} tells us that we converge to some $\bargamma^*$ such that $\sT \bargamma^* = \bargamma^*$, i.e.
we reach an exact UD solution. (Note that Lemma \ref{lemma:ER_preserves_stability} gives stability of the iterates,
and stability of the limit point $\bargamma^*$ follows.)
As a corollary, we recover Theorem \ref{thm:UD_existence} on existence of UD solutions.
\end{remark}

\section{Stability is Critical}
\label{sec:stability_critical}

This section demonstrates that our approach of starting with a stable allocation, and ensuring that stability is preserved,
plays a critical
role in our construction of an FPTAS using iterative edge rebalancing.

Let $n \equiv |V|$. Appendix \ref{app:example_stability_critical} shows the following.
There is a sequence of instances $(I_n, n \geq 8)$, such that for
each instance in the sequence the following holds. (a) The instance admits a UD solution.
(b) There is an outcome $(\bargamma, M^*)$ on a maximum weight
matching $M^*$ such that:
\begin{enumerate}
\item The outcome satisfies $\eps$-correct division for $\eps =  2^{-cn}$.
\item (Stability violation) There is a `bad' edge $(i,j) \notin M^*$ such that $\gamma_i + \gamma_j \leq w_{ij} - 1$
\end{enumerate}
where $c>0$ is a constant. Split fractions are bounded within $[r, 1-r]$ for arbitrary desired
$r \in (0, 1/2)$ ($c$ depends on $r$). Also, the weights are uniformly bounded by a constant $W(r)$.

We now describe the implications of such a construction. 
Suppose we perform edge balancing on the
example outcome (as per Eq.~\eqref{eq:mann_iterations}, using operator $\sT$ defined there),
i.e. $\bargamma^0 \equiv \bargamma$.
We know that $\lVert \sT \bargamma^0 - \bargamma^0 \rVert_\infty \leq \eps$,
since $\bargamma^0$ satisfies $\eps$-correct
division. Define $\sT_\damp \equiv \damp \sT + (1-\damp)\sI$, where $I$ is the identity operator.
Eq.~\eqref{eq:mann_iterations} simply corresponds to iterating with $\sT_\damp$, i.e. $\bargamma^t = \sT_\damp^t \bargamma^0$.
Clearly, $\lVert \sT_\damp \bargamma^0 - \bargamma^0 \rVert_\infty \leq \eps$.
Also, it follows from non-expansivity of $\sT$ (as per Eq.~\eqref{eq:T_non_expansive}) that $\sT_\damp$ is non-expansive
in sup norm. As a consequence
$\lVert \sT_\damp \bargamma^t - \bargamma^t \rVert_\infty \leq \eps$ for all $t \geq 0$.
 Thus, successive iterates differ by at most $\eps$ in sup norm, meaning
 that no coordinate changes by more than $\eps$ per iteration. Suppose we want to reach a
 configuration that satisfies both $(1/2)$-stability ($\gamma_k + \gamma_l \geq w_{kl} - 1/2$ for each $(k,l) \in E$)
 and the $(1/2)$-correct division condition.
One of $\gamma_i$ and $\gamma_j$ must change by at least $1/4$ for
the `bad' edge $(i,j)$ to satisfy $(1/2)$-stability,
i.e. $\gamma_i + \gamma_j \geq w_{ij} - 1/2$. But this will take at least
$1/(4\eps)=2^{\Omega(n)}$ iterations!

Thus, \emph{it can take exponential time to reach an approximate UD solution if we do
not stay within the space of stable outcomes while rebalancing.}

\begin{remark}
Essentially the same construction and reasoning shows that the dynamics
of  \cite{OurNewArxiv} can take exponential time to reach an $\eps$-UD solution.
\end{remark}

\paragraph{Further directions.}

It remains open whether there is a polynomial
algorithm that finds an exact UD solution.

Second, it would be interesting to identify other classes of games where solution concepts that are not symmetrical in the players
can be naturally defined and motivated. Various classes of cooperative games seem like particularly suitable candidates.

Third, though we have found a fast local algorithm  for finding $\eps$-UD solutions, it does not constitute a natural
description of market behavior of the type proposed in \cite{OurNewArxiv}. However, as discussed in Section \ref{sec:stability_critical},
there are instances where that dynamics does \textit{not} quickly reach a solution in the unsymmetrical case.
So it is unclear how to resolve this question.
\vspace{0.05cm}

\newpage

{\bf\large Acknowledgements.}
The author would like
to thank Andrea Montanari, Mohsen Bayati, R. Ravi and Mohammad Hossein Bateni for helpful discussions.

\bibliographystyle{amsalpha}

\newpage

\appendix

\section{Appendix to Section \ref{subsec:local_algorithm}}
\label{app:fast_local}

\begin{remark}
\label{rem:unique_opt_isgeneric}
Fix a graph $G=(V,E)$ and maximum weight $W>0$.
We argue that the condition ``LP \eqref{prob:mwm_relaxation} has a unique optimum" is generic in each of two different cases:
\begin{enumerate}[(i)]
\item {\bf All instances:} Let ${\sf G}=(0,W]^{|E|}$ be the set of all instances.
Then the subset of instances with unique optimum is both
\emph{open and dense in ${\sf G}$}.
\item {\bf Instances with integral optimum:}
Let ${\sf G}_{\textup{I}} \subset (0,W]^{|E|}$ be the set of instances having an integral optimum.
Let ${\sf G}_{\textup{UI}} \subset {\sf G}_{\textup{I}}$ be the set of instances having a unique integral optimum.
It turns out that ${\sf G}_{\textup{I}}$ has dimension $|E|$ (i.e.
the class of instances having an integral optimum is large) and
that ${\sf G}_{\textup{UI}}$ is both
\emph{open and dense in ${\sf G}_{\textup{I}}$}.
\end{enumerate}
\end{remark}

\begin{proof}[Proof of Proposition \ref{propo:BP_gives_dual_opt}]
Take any edge $(i,j) \in M^*$. 
From Lemma \ref{lemma:BP_bayati_etal}, we know that $\off{j}{i}^* > 0$. It follows that
$\off{j}{i}^* = w_{ij} - \max_{l \in \partial j \backslash i} \off{l}{i}^*$. But $\off{j}{i}^* = \mu_i(1)$ and
$\max_{l \in \partial j \backslash i} \off{l}{i}^* = \mu_j(2)$ by Lemma \ref{lemma:BP_bayati_etal}. Thus we obtain
\begin{align}
\mu_i(1) = w_{ij} - \mu_j(2)
\label{eq:mu_i1_j2}
\end{align}
Similarly, we have
\begin{align}
\mu_j(1) = w_{ij} - \mu_i(2)
\label{eq:mu_i2_j1}
\end{align}
Combining Eq.~\eqref{eq:mu_i1_j2},\eqref{eq:mu_i2_j1}, we obtain $y_i^* + y_j^* = w_{ij}$ as required.

Take any edge $(i,j) \notin M^*$. 
From Lemma \ref{lemma:BP_bayati_etal},
we know that $\off{i}{j} \leq \mu_j(2)$. Also, $\max_{k \in \partial i \backslash j} \off{k}{i}= \mu_i(1)$.
It follows that
$(w_{ij}- \mu_i(1) )_+ \leq \mu_j(2) \Rightarrow w_{ij} \leq \mu_i(1) + \mu_j(2)$.
Similarly, we obtain
$w_{ij} \leq \mu_j(1) + \mu_i(2)$.
Combining, it follows that $y_i^* + y_j^* \geq w_{ij}$ as required.

Note also that for any $i\in V$ not matched under $M^*$, it follows from Lemma \ref{lemma:BP_bayati_etal} that
$y_i^*=0$. Thus, we have shown that $\bary^*$ is a feasible point for the dual LP \eqref{prob:mwm_dual}, which
also satisfies $\sum_{i \in V} y_i^* = \textup{weight of } M^*$,
i.e. it achieves the value of the primal LP \eqref{prob:mwm_relaxation}.
Hence, $\bary^*$ is a optimum solution to the dual LP \eqref{prob:mwm_dual}.
\end{proof}

\section{Local polynomial time `reduction' to maximum weight matching}
\label{app:dualopt_from_Mstar}

In this section we prove the following:
\begin{claim}
\label{claim:Mstar_to_eps_UD}
Given a maximum weight matching $M^*$ for an instance possessing a UD solution, an $\eps$-UD solution
can be constructed by a local algorithm with computational effort $\poly(|V|, 1/\eps)$.
\end{claim}

Our definition of $\eps$-UD solutions retains a strict version of stability while relaxing the
balance requirement to $\eps$ balance (cf. Definition \ref{def:UD_soln}). We use max-product belief propagation
to find a stable allocation, \emph{given a maximum weight matching $M^*$}. This is achieved locally and in polynomial time.


Consider the standard \emph{undamped synchronous BP updates} given by:
\begin{align}
\off{i}{j}^t &= (w_{ij}- \alf{i}{j}^t)_+\nonumber\\[2pt]
\alf{i}{j}^{t+1} &=  \max_{k \in \partial i \backslash j} \off{k}{i}^t
\label{eq:undamped_BP}
\end{align}
This is equivalent to the update rule Eq.~\eqref{eq:BP_one_step_update}.

We use a carefully chosen initialization (different from the usual all-zero) to achieve our objective:
\begin{align}
\off{i}{j}^0 = \left \{\begin{array}{ll}
w_{ij} & \mbox{if }(ij) \in M^*\\
0 & \mbox{otherwise}
\end{array} \right .
\label{eq:BP_special_initialization}
\end{align}

Let the version of max-product BP message passing defined by \eqref{eq:undamped_BP} and \eqref{eq:BP_special_initialization} be denoted by $\mc{A}$.

Our key result on $\mc{A}$ is the following:
\begin{claim}
Algorithm $\mc{A}$ converges to an exact fixed point in $2|E|$ iterations.
\label{claim:BP_givenMstar_converges}
\end{claim}

The fixed points of update rule $\eqref{eq:undamped_BP}$ can be characterized similarly to the fixed points of the
`natural dynamics' in our previous work \cite[Section 3]{OurNewArxiv}.
\begin{lemma}
\label{lemma:BP_fixed_point}
Consider an instance having an integral optimum to LP \eqref{prob:mwm_relaxation}, corresponding
to matching $M^*$.
The update rule \eqref{eq:undamped_BP} has at least  one fixed point.
Let $(\baralf^*,\baroff^*)$ be a fixed point. Then
\begin{align*}
w_{ij}- \alf{i}{j}^*-\alf{j}{i}^* \geq 0 \quad \forall \; (ij) \in M^*\\
w_{ij}- \alf{i}{j}^*-\alf{j}{i}^* \leq 0 \quad \forall \; (ij) \notin M^*
\end{align*}
Also, for every $(ij) \notin M^*$, we have $\alf{i}{j}^* = \off{k}{i}^*$, where $(i,k) \in M^*$.
\end{lemma}
The above lemma follows directly from the arguments in (\cite{OurNewArxiv}, Appendix F.1).

We now show how Claim \ref{claim:BP_givenMstar_converges} implies Claim \ref{claim:Mstar_to_eps_UD}.

\begin{proof}[Proof of Claim \ref{claim:Mstar_to_eps_UD}]
Using Lemma \ref{lemma:BP_fixed_point}, we can show
that for any fixed point $\baroff^*$ of Eq.~\eqref{eq:undamped_BP} for an instance such that LP \eqref{prob:mwm_relaxation} has
an integral optimum $M^*$, the following holds: For any $i \in V$, if $i$ is matched under $M^*$ to $j$ then
$j \in \argmax_{k \in \partial i} \off{k}{i}^*$ and $\off{j}{i}^*= w_{ij} - \alf{j}{i}^*$, whereas
if $i$ is unmatched under $M^*$ then $\off{k}{i}^*=0$ for every $k \in \partial i$.
We use the construction for $\bary^*$ described
in Section \ref{subsec:local_algorithm}, and essentially the same proof of Proposition \ref{propo:BP_gives_dual_opt} goes through.
A stable allocation  $(\bary^*, M^*)$ thus follows from any fixed point $\baroff^*$ of $\mc{A}$, where
$M^*$ is the given maximum weight matching.

Starting from a stable allocation, an $\eps$-UD solution can be constructed with effort $O(|E|/\eps^2)$
using {\sc Edge Balancing} as described in
Section \ref{sec:main_results}. The claim follows.
\end{proof}


We devote the rest of this subsection to the proof of Claim \ref{claim:BP_givenMstar_converges}.

Next, we define a useful partial ordering on message vectors $\baroff$.
\begin{definition}
\label{def:BP_givenMstar_ordering}
We say $\baroff \preceq \barsoff$ if the following hold:
\begin{align*}
\off{i}{j} \geq \soff{i}{j} & \qquad \forall \;(ij) \in M^*\\
\off{i}{j} \leq \soff{i}{j} & \qquad \forall \;(ij) \notin M^*
\end{align*}
\end{definition}

\begin{lemma}
\label{lemma:BP_A_monotonicity}
Let $\baroff^*$ be a fixed point of update rule \eqref{eq:undamped_BP}. Algorithm $\mc{A}$ satisfies
\begin{align}
\baroff^t \preceq \baroff^{t+1} \preceq \baroff^* \qquad \forall \; t \geq 0 \, .
\label{eq:A_monotonicity}
\end{align}
Also, for all $(ij) \notin M^*$, we have
\begin{align*}
\alf{i}{j}^{t+1} = \off{k}{i}^t \qquad \forall \; t \geq 0 \, ,
\end{align*}
where $(i,k) \in M^*$.
\end{lemma}
\begin{proof}
It is trivial to check validity for $t=0$. Suppose the result is true up to $t-1$. Take any fixed point $\baroff^*$. Then
\begin{align}
\baroff^t \preceq \baroff^*
\label{eq:BPoff_lt_fixedoff}
\end{align}
Now consider any unmatched edge $(ij) \notin M^*$. If $i$ is matched, say $(i,k) \in M^*$, it follows from Eq.~\eqref{eq:BPoff_lt_fixedoff} and Lemma \ref{lemma:BP_fixed_point} that
$\alf{i}{j}^{t+1} = \off{k}{i}^t$ as needed. Further, $\off{k}{i}^t \leq \off{k}{i}^{t-1} = \alf{i}{j}^t$, leading to
\begin{align*}
&\alf{i}{j}^{t+1} \leq \alf{i}{j}^t \\
\Rightarrow \hskip20pt &\off{i}{j}^{t+1} \geq \off{i}{j}^t.
\end{align*}
as needed. Else if $i$ is unmatched under $M^*$, $\off{k}{i}^t \leq \off{k}{i}^* = 0$ for all $k \in \partial i$. Hence,
$\alf{i}{j}^{t+1}= \alf{i}{j}^t = 0$, leading to $\off{i}{j}^{t+1} = \off{i}{j}^t =w_{ij}$. This suffices.

On the other hand, for every matched edge $(i,k)$,
\begin{align*}
&\alf{i}{k}^{t+1} = \max_{j \in \partial i \backslash k} \off{j}{i}^t \geq \max_{j \in \partial i \backslash k} \off{j}{i}^{t-1} = \alf{i}{k}^{t}\\
\Rightarrow \hskip20pt &\hskip60pt\off{i}{k}^{t+1} \leq \off{i}{k}^t.
\end{align*}
as needed.

The second inequality $\baroff^{t+1} \preceq \baroff^*$ can be established similarly.

Induction completes the proof.
\end{proof}
Next, we present a key construction leading to a proof of Claim \ref{claim:BP_givenMstar_converges} for the case that LP \eqref{prob:mwm_relaxation} has a unique optimum:

Choose a fixed point $\baroff^*$. A \emph{critical path} $P$ leading to a message $\off{i_1}{i_0}^*$  is constructed as follows:\\
\vspace{0.3cm}

\begin{tabular}{ll}
\hline
\multicolumn{2}{l}{ {\sc Critical Path}( Instance $G$, BP-fixed point $\baroff^*$)}\\
\hline
1: & $k\leftarrow 1$ \\
2: & While $\off{i_k}{i_{k-1}}^*>0$\\
3: & \hspace{0.2cm} Find $i_{k+1} \in \argmax_{j \in \partial i_k \backslash i_{k-1}} \off{j}{i_1}^*$.\\
4: & \hspace{0.2cm} $k \leftarrow k+1$\\
5: & \hspace{0.2cm} If $(i_{k},i_{k-1}) = (i_{l},i_{l-1})$ for some $l<k$ then \\
6: & \hspace{0.4cm} Break While; \\
7: & \hspace{0.2cm} End If \\
8: & End While\\
9: & Return $(i_k, i_{k-1}, \ldots, i_0)$\\
\hline
\end{tabular}

\vspace{0.3cm}
It is easy to see that a critical path can have at most $2|E|+1$ directed edges since the path is terminated if a directed edge repeats, and there are $2|E|$ distinct directed edges in the graph.



\begin{proof}[Proof of Claim \ref{claim:BP_givenMstar_converges}: Unique LP optimum case]

Take a fixed point $\baroff^*$. Consider any edge $(i_1, i_0)$. Let its critical path be $(i_k, i_{k-1}, \ldots, i_0)$. There are two cases:\\
Case (i): $\off{i_k}{i_{k-1}}^*=0$\\
This is the case where no directed edge repeats. In this case, we claim that $\off{i_1}{i_0}^{k-1}= \off{i_1}{i_0}^*$. We simply start with the evident $\off{i_k}{i_{k-1}}^0= \off{i_k}{i_{k-1}}^*$ and move sequentially along the critical path. Lemmas \ref{lemma:BP_fixed_point} and \ref{lemma:BP_A_monotonicity} ensure that $\alf{i_{k-1}}{i_{k-2}}^1=\off{i_k}{i_{k-1}}^0$, $\alf{i_{k-2}}{i_{k-3}}^2=\off{i_{k-1}}{i_{k-2}}^1$ and so on. The monotonicity established plays a key role here. This leads to $\off{i_{k-1}}{i_{k-2}}^1=\off{i_{k-1}}{i_{k-2}}^*$, $\off{i_{k-2}}{i_{k-3}}^2=\off{i_{k-2}}{i_{k-3}}^*$ and so on, leading to the result.
\item Case (ii): $(i_k,i_{k-1}) = (i_l,i_{l-1})$ for some $l<k$\\
In this case we know that all messages along the critical path are strictly positive on unmatched edges, and hence are not thresholded at 0. It follows, going around the directed alternating cycle, $\mathcal{C}=(i_k, \ldots, i_l)$ that for any directed edge $(\ell,\ell-1) \in \mathcal{C}$
\begin{align*}
\off{\ell}{\ell-1}^* \leq \off{\ell}{\ell-1}^* + \mbox{wt. of unmatched edges in }\mc{C} - \mbox{wt. of matched edges in }\mc{C} < \off{\ell}{\ell-1}^*
\end{align*}
since we have assumed that LP \eqref{prob:mwm_relaxation} has a unique solution. This is a contradiction.

Thus, Case (ii) never arises. Case (i) implies that $k \leq 2|E|$. Hence all messages converge to values at $\baroff^*$ in $2|E|$ iterations.
\end{proof}

Note that the proof above implies that the fixed point $\baroff^*$ is unique!

\begin{proof}[Sketch of proof of Claim \ref{claim:BP_givenMstar_converges}: non-unique LP optimum]

Claim \ref{claim:BP_givenMstar_converges} holds also for the case where LP \eqref{prob:mwm_relaxation} is tight but not pointed. We only briefly sketch the proof in this case. In the non-unique optimum case, max product may have multiple fixed points. However, there is a unique smallest fixed point $\baroff^{*,\bot}$ with respect to the partial ordering defined in \ref{def:BP_givenMstar_ordering}, and algorithm $\mathcal{A}$ converges to this fixed point by monotone convergence (cf. Eq.~\eqref{eq:A_monotonicity}).

In fact, we can show that the same bound $2|E|$ holds on the time to convergence. To prove this we compare against the special fixed point $\baroff^{*,\bot}$, whose minimality plays a crucial role. We use a similar critical path construction as for the unique optimum case. However, we have to be more careful here: we break ties in selecting an element of $\argmax_{j \in \partial  i_k \backslash i_{k-1}} \off{j}{i_1}^*$ non-deterministically (\emph{there exists} a sequence of tie-breaking choices such that ...). The same cases (i) and (ii) arise in the proof of message convergence using the critical path (cf. proof for unique LP optimum). Case (i) goes through as before. For Case (ii), we use the minimality of $\baroff^{*,\bot}$ to arrive at a contradiction.
\end{proof}

\section{An Example showing that Stability is Critical}
\label{app:example_stability_critical}

Let $n=|V|$. In this section, we construct a sequence of instances $(I_n, n \geq 8)$, such that for
each instance in the sequence the following holds. (a) The instance admits a UD solution.
(b) There is an outcome $(\bargamma, M^*)$ on a maximum weight
matching $M^*$ such that:
\begin{enumerate}
\item The outcome satisfies $\eps$-correct division for $\eps =  2^{-cn}$.
\item (Stability violation) There is a `bad' edge $(i,j) \notin M^*$ such that $\gamma_i + \gamma_j \leq w_{ij} - 1$
\end{enumerate}
where $c>0$ is a constant. Split fractions are bounded within $[r, 1-r]$ for arbitrary desired
$r \in (0, 1/2)$ ($c$ depends on $r$). Also, the weights are uniformly bounded by a constant $W(r)$.

Such a construction implies that we cannot hope to converge in worst case polynomial time to an approximate UD
solution, if we start the rebalancing process (cf. Table \ref{alg:edge_reb}) with an arbitrary allocation
corresponding to $M^*$. This is discussed in Section \ref{sec:stability_critical}. Thus, our strategy
of staying within the space of \emph{stable} configurations plays  a critical role.

We now define the instance $I_n$. Let us first consider $n= 8 N$, where $N \in \mathbb{Z}$. Later we show how to
extend the construction to arbitrary $n \geq 8$. The graph $G_n=(V_n, E_n)$ we will consider is a simple `ring'. More precisely,
$V_n=\{1, 2, \ldots, n\}$ and $E_n=\{(1,2), (2,3), \ldots, (n-1,n), (n,1) \}$. All edges have the same weight $W$.
This graph has two integral maximum weight matchings (cf. Remark \ref{rem:multiple_mwm_not_needed} below),
we pick $M^*= \{ (1,2), (3,4), \ldots, (n-1,n) \}$.
Given any $r \in (0, 1/2)$,
we define the split fractions as follows:
\begin{align*}
r_{1,2} &= r_{3,4} = \ldots = r_{2N-1, 2N} = r \\
r_{4N,4N-1} &= r_{4N-2,4N-3} = \ldots = r_{2N+2, 2N+1} = r
\end{align*}
Note that the values of the split fractiom on the edges $(2,3), (4,5), \ldots \notin M^*$ are irrelevant,
given our choice of matching.
As before, $r_{i,i+1}= 1- r_{i+1,i}$ is implicit.

For $l> 4N$, we define split fractions in a symmetrical way. Define `reflection' $\rfl : \{ 4N+1, 4N+2, \ldots, 8N\} \rightarrow \{ 1, 2, \ldots, 4N \}$
as
\begin{align}
\rfl (l) = 8N - l + 1
\end{align}
We set $r_{2i+1, 2i+2} = r_{\rfl(2i+1), \rfl(2i+1)}$ for all $i \in \{ 2N, 2N+1, \ldots, 4N-1\}$.

Note that the allocation in which each node earns $W/2$, together with matching $M^*$,
constitutes a UD solution for the instance defined.

Now we show how to construct an outcome $(\bargamma, M^*)$ satisfying properties 1 and 2 above.
Let $\beta \equiv (1-r)/r > 1$. Define
$$\eps' = \frac{1}{\beta^{N-1}} \,. $$
For $0 \leq i \leq N-1$ we choose
\begin{align}
\gamma_{2(N-i)} = \gamma_{2(N+i)+1} = \frac{W}{2}+ \frac{1}{2} + \frac{1-\beta^{-i}}{\beta - 1}
\end{align}
In each case $\gamma_{2j} = W - \gamma_{2j-1}$, since we want a valid outcome. Thus, we have defined $\gamma_1, \gamma_2, \ldots, \gamma_{4N}$.
The remaining earnings are defined as,
\begin{align}
\gamma_{i} = W- \gamma_{\rfl(i)} \qquad \mbox{for } i = 4N+1, 4N+1, \ldots, 8N
\end{align}
It is  easy to see that this definition satisfies the fixed sum constraints on all edges in $M^*$.

Importantly, note that it suffices to have $W \geq 1 + 2/(\beta-1)$ to ensure that this is a valid allocation with each
$\gamma_i \in (0,W)$. Choose $W\equiv 1 + 2/(\beta-1)$ (for example).

See that $\gamma_{6N}+\gamma_{6N+1} = W -1$ thus satisfying property 2. We show next that $(\bargamma, M^*)$ satisfies
the $\eps'$-correct division condition.

Consider the edge $(2(N-i), 2(N-i)+1)$ for $1\leq i \leq N-1$. It is easy to see that
$\gamma_{2(N-i)} + \gamma_{2(N-i)+1)} = W + \beta^{-i}$. Also, $\gamma_{2N}+\gamma_{2N+1} = W+1$.
It follows from a short calculation that
the exact `correct division' requirement Eq.~\eqref{eq:correct_division} is satisfied by the matched edges
$(3,4), (5,6), \ldots , (2N-1, 2N)$. For the matched edge $(1,2)$, note that
$\gamma_1+ \gamma_{8N} = W$, whereas $\gamma_2+ \gamma_3 = W+\beta^{-(N-1)} = W+\eps'$.
It follows that edge $(1,2)$ satisfies $\eps'$-correct division. Similar arguments take care of all
the other matched edges in the other three `quarters' of the ring (in fact, the argument can be completed
using symmetry). Thus, we have an outcome satisfying $\eps'$-correct division (cf. Definition \ref{def:eps_correct_div}).
Since $\eps' = \beta^{-(N-1)}$ it follows that property 1 above is satisfied provided $c$ is chosen
appropriately.

For $n \geq 8$, but not a multiple of 8, we simply use the construction above for $n' = 8 \lfloor n/8 \rfloor$ and
add a dummy component of size $n - n'$, disconnected from $G_{n'}$. Further, we fix a UD solution on the dummy component (any bipartite
graph has a UD solution).
Since $n' \geq n/2$, it follows that property 1 is satisfied
if $c$ is chosen appropriately. ($c\equiv (1/4)\log_2 \beta$ works for all $n \geq 8$.) Clearly, property 2 is also satisfied.

Note that it was only for the sake of simplicity that the example we gave had multiple maximum weight matchings.
\begin{remark}
\label{rem:multiple_mwm_not_needed}
Though the example constructed above has multiple maximum weight matchings, this is not necessary.
We can in fact, construct examples (that admit UD solutions) with the same properties 1 (for appropriate $c>0$) and 2 above, and $W= O(1)$,
where the weight of the maximum weight matching is at least $1$ more than the weight of the next heaviest
matching.
\end{remark}

%
\end{document}